\documentclass[twoside,11pt]{article}

\usepackage{natbib}
\usepackage{amsmath}
\usepackage{amsthm}
\usepackage{amsfonts}
\usepackage{amssymb}
\usepackage{graphicx}
\usepackage[a4paper, margin = .9in]{geometry}
\usepackage{mathptmx}
\usepackage{hyperref}
\usepackage{dsfont}
\usepackage{algorithm}
\usepackage{algorithmicx}
\usepackage{algpseudocode}
\usepackage{subfigure}

\newtheorem{theorem}{Theorem}
\newtheorem{definition}{Definition}
\newtheorem{remark}{Remark}
\newtheorem{lemma}{Lemma}

\begin{document}

\title{A note on quickly sampling a sparse matrix with low rank expectation}

\author{Karl Rohe, Jun Tao, Xintian Han, Norbert Binkiewicz    }
       
%

\maketitle

\begin{abstract}
Given matrices $X,Y \in \mathds{R}^{n \times K}$ and $S \in \mathds{R}^{K \times K}$ with positive elements, this paper proposes an algorithm \texttt{fastRG} to sample a sparse matrix $A$ with low rank expectation $E(A) = XSY^T$ and independent Poisson elements.  This allows for quickly sampling from a broad class of stochastic blockmodel graphs (degree-corrected, mixed membership, overlapping) all of which are specific parameterizations of the generalized random product graph model defined in Section \ref{gRPG}. The basic idea of \texttt{fastRG} is to first sample the number of edges $m$ and then sample each edge. The key insight is that because of the the low rank expectation, it is easy to sample individual edges.  The naive ``element-wise'' algorithm requires $O(n^2)$ operations to generate the $n\times n$ adjacency matrix $A$.  In sparse graphs, where $m = O(n)$, ignoring log terms, \texttt{fastRG} runs in time $O(n)$.  An implementation in \texttt{R} is available on github.  A computational experiment in Section \ref{sec:implementation} simulates graphs up to $n=10,000,000$ nodes with $m = 100,000,000$ edges.  For example, on a graph with $n=500,000$ and $m = 5,000,000$, \texttt{fastRG} runs in less than one second on a 3.5 GHz Intel i5.  
\end{abstract}

%

\section{Introduction}
The random dot product graph (RDPG) model serves as a test bed for various types of clustering and statistical inference algorithms. 
This model generalizes the Stochastic Blockmodel (SBM), the Overlapping SBM, the Mixed Membership SBM, the Degree-Corrected SBM, and the latent eigenmodel which is itself a generalization of the Latent Space Model \citep{Holland1983Stochastic,
Latouche2011Overlapping,
Airoldi2008Mixed,
PhysRevE.83.016107,
Hoff2007Modeling,hoff2002latent}. Under the RDPG, each node $i$ has a (latent) node feature $x_i\in \mathds{R}^K$ and the probability that node $i$ and $j$ share an edge is parameterized by $\langle x_i , x_j \rangle$ \citep{Young2007}. 

While many network analysis algorithms only require $O(m)$ operations, where $m$ is the number of edges in the graph, sampling RDPGs with the naive ``element-wise'' algorithm takes  $O(n^2)$ operations, where $n$ is the number of nodes.  In particular, sparse eigensolvers compute the leading $k$ eigenvectors of such graphs in $O(km)$ operations (e.g. in ARPACK \citep{lehoucqarpack}).  As such, sampling an RDPG is a computational bottleneck in simulations to examine many network analysis algorithms.


The organization of this paper is as follows.  Section \ref{fastRG} gives \texttt{fastRG}. Section \ref{sec:xlr} relates \texttt{fastRG} to xlr, a new class of edge-exchangeable random graphs with low-rank expectation.  Section \ref{gRPG} presents a generalization of the random dot product graph. Theorem \ref{theoremRDPG} shows that \texttt{fastRG} samples Poisson-edge graphs from this model.  Then, Theorem \ref{theoremThresh} in Section \ref{approximate} shows how \texttt{fastRG} can be used to approximate a certain class of Bernoulli-edge graphs.  Section \ref{sec:implementation} describes our implementation of \texttt{fastRG} (available at \url{https://github.com/karlrohe/fastRG}) and assesses the empirical run time of the algorithm.


\subsection{Notation}
Let $G = (V,E)$ be a graph with the node set $V = \{1, \ \dots,\ n\}$ and the edge set $E = \{(i,j): \mbox{$i$ is connected to $j$}\}$. In a directed graph, each edge is an ordered pair of nodes while in an undirected graph, each edge is an unordered pair of nodes. A multi-edge graph allows for repeated edges. In any graph, a self-loop is an edge that connects a node to itself. Let the adjacency matrix $A \in \mathds{R}^{n\times n}$ contain the number of edges from $i$ to $j$ in element $A_{ij}$.  For column vectors $x \in \mathds{R}^a$, $z \in \mathds{R}^b$ and $S \in \mathds{R}^{a \times b}$, define $\langle x, z \rangle _S = x^T S z$;  this function is not necessarily a proper inner product because it does not require that $S$ is non-negative definite. 
We use standard big-$O$ and little-$o$ notations,
i.e. for sequence $x_n$, $y_n$; $x_n = o(y_n)$ when $y_n$ is nonzero implies $\lim_{n\rightarrow \infty}x_n/y_n=0$; $x_n = O(y_n)$ implies there exists a positive real number $M$ and an integer $N$ such that $|x_n| \leq M|y_n|, ~\forall n \geq N.$

\section{\texttt{fastRG}}
\label{fastRG}
\texttt{fastRG} is motivated by the 
 wide variety of low rank graph models that specify the expectation of the adjacency matrix as $E(A) = XSX^T$ for some matrix (or vector) $X$ and some matrix (or value) $S$.

\begin{table}[htbp]
   \centering
\begin{tabular}{|l|l|l|}
\hline Types of low rank models & $X \in$  & In each row...\\
\hline SBM & $\{0,1\}^{n\times K}$ & a single one  \\
Degree-Corrected SBM & $\mathds{R}^{n \times K}$ & a single non-zero positive entry \\
Mixed Membership SBM & $\mathds{R}^{n\times K}$ & non-negative and sum to one \\
Overlapping SBM & $\{0,1\}^{n \times K}$ & a mix of $1$s and $0$s \\
Erd\H{o}s-R\'enyi  & $\{1\}^n$ &  a single one\\
Chung-Lu  & $\mathds{R}^n$ &  a single positive value\\
\hline \end{tabular}
   \caption{Restrictions on the matrix $X$ create different types of well known low rank models.  
   There are further differences between these models that are not emphasized by this table.}
   \label{tab:blockmodels}
\end{table}

Given $X \in \mathds{R}^{n \times K_x}$, $Y \in \mathds{R}^{d \times K_y}$ and $S \in \mathds{R}^{K_x \times K_y}$,
 \texttt{fastRG} samples a random graph whose $n\times d$ incidence matrix has expectation $XSY^T$.  Importantly, \texttt{fastRG} requires that the elements of $X, Y,$ and $S$ are non-negative.  This condition holds for all of the low rank models in the above table. 
Each of those models set $Y = X$ and enforce different restrictions on the matrix $X$.

As stated below, \texttt{fastRG} samples a (i) directed graph with (ii) multiple edges and (iii) self-loops. After sampling, these properties can be modified to create a simple graph (undirected, no repeated edges, and no self-loops);  see Remarks \ref{remark:undirected} and \ref{remark:selfloops} in Section \ref{gRPG} and Theorem \ref{theoremThresh} in Section \ref{approximate}.

\begin{algorithm}
\caption{\texttt{fastRG}($X,S,Y$)}
    \begin{algorithmic}
        \Require $X \in \mathds{R}^{n \times K_x}$, $S \in \mathds{R}^{K_x \times K_y}$, and $Y \in \mathds{R}^{d \times K_y}$  with all matrices containing non-negative entries.
            \State Compute diagonal matrix $C_X \in \mathds{R}^{K_x \times K_x}$ with $C_X=diag(\sum_i X_{i1} \ , \  \dots \ , \ \sum_i X_{iK_x})$.
            \State Compute diagonal matrix $C_Y \in \mathds{R}^{K_y \times K_y}$ with $C_Y=diag(\sum_i Y_{i1} \ , \  \dots \ , \ \sum_i Y_{iK_y})$.
            \State Define $\tilde X = X C_X^{-1}$,  $\tilde S = C_X S C_Y$, and $\tilde Y = Y C_Y^{-1}$.  
            \State Sample the number of edges $m \sim Poisson(\sum_{u,v} \tilde S_{uv})$.
            \For {$\ell = 1:m $}
                \State Sample $U \in \{1,\ \dots,\ K_x\}, V \in \{1,\ \dots,\ K_y\}$ with $\mathbb{P}(U = u, V=v) \propto \tilde S_{uv}$.
                \State Sample $I \in \{1,\ \dots,\ n\}$ with $\mathbb{P}(I=i) = \tilde X_{iU}$.
                \State Sample $J \in \{1,\ \dots,\ d\}$ with $\mathbb{P}(J=j) = \tilde Y_{jV}$.
                \State Add edge $(I,J)$ to the graph, allowing for multiple edges $(I,J)$.
            \EndFor
    \end{algorithmic}
\end{algorithm}

An implementation in \texttt{R} is available at \url{https://github.com/karlrohe/fastRG}.  As discussed in Section \ref{sec:implementation}, in order to make the algorithm more efficient, the implementation is slightly different from the statement of the algorithm above.

There are two model classes that can help to interpret the graphs generated from \texttt{fastRG} and those model classes are explored in the next two subsections.  Throughout all of the discussion, the key fact that is exploited by \texttt{fastRG} is given in the next Theorem.

\begin{theorem}\label{theorem:xlr}
Suppose that $X \in \mathds{R}^{n \times K_x}$, $Y \in \mathds{R}^{d \times K_y}$ and $S \in \mathds{R}^{K_x \times K_y}$  all contain non-negative entries. Define $x_i \in \mathds{R}^{K_x}$ as the $i$th row of $X$. Define $y_j \in \mathds{R}^{K_y}$ as the $j$th row of $Y$.  Let $(I,J)$ be a single edge sampled inside the for loop in \texttt{fastRG}$(X,S,Y)$, then
\[\mathbb{P}\left((I,J) = (i,j)\right) \propto \langle x_i,y_j \rangle _S.\]
\end{theorem}
\begin{proof}
\begin{eqnarray*} 
\mathbb{P}\big((I,J)=(i,j)\big) &=& \sum_{u,v} \mathbb{P}((I,J)=(i,j)|(U,V)=(u,v)) \ \mathbb{P}((U,V)=(u,v))) \\
&=& \frac{\sum_{u,v} \tilde X_{iu} \tilde Y_{jv}\tilde S_{uv}}{\sum_{u,v} \tilde S_{u,v}} = \frac{ \sum_{u,v} X_{iu} Y_{jv} S_{uv}}{\sum_{u,v} \tilde S_{u,v}} = \frac{x_i^TS y_j}{\sum_{a,b}  x_a^TS y_b}
\end{eqnarray*}
\end{proof}

\noindent
 For notational simplicity, the rest of the paper will suppose that $Y = X \in \mathds{R}^{n \times K}$.

\subsection{\texttt{fastRG} samples from xlr: a class of edge-exchangeable random graphs}\label{sec:xlr}

There has been recent interest in edge exchangeable graph models with blockmodel structure (e.g. \cite{NIPS2016_6521, 2016arXiv160202114T}). To characterize a broad class of such models, we propose xlr.

\begin{definition}\label{def:xlr}[xlr] An xlr graph on $n$ nodes and $K$ dimensions is generated as follows,
\begin{enumerate}
\item Sample $(X,S) \in   \mathds{R}^{n \times K} \times \mathds{R}^{K \times K}$ from some distribution and define $x_i$ as the $i$th row of $X$.
\item Initialize the graph to be empty.
\item Add independent edges $e_1, e_2, \dots$ to the graph, where 
\[\mathbb{P}(e_\ell = (i,j)) = \frac{\langle x_i,x_j \rangle _S }{ \sum_{a,b} \langle x_a,x_b \rangle _S}.\]
\end{enumerate}
\end{definition}
From Theorem \ref{theorem:xlr}, \texttt{fastRG} provides a way to sample the edges in xlr.

An xlr graph is both (i) edge-exchangeable as defined by \cite{crane2016edge} and (ii) conditional on $X$ and $S$, its expected adjacency matrix is low rank.  By sampling $X$ to satisfy one set of restrictions specified in Table \ref{tab:blockmodels}, xlr provides a way to sample edge exchangeable blockmodels.    xlr stands for \textit{edge-exchangeable and low rank} because it characterizes all edge-exchangeable and low rank random graph models on a finite number of nodes.  
In particular, by Theorem 4.2 in \cite{crane2016edge}  if a random undirected graph with an infinite number of edges is edge exchangeable, then the edges are drawn iid from some randomly chosen distribution on edges $f$.  Moreover, let $B$ be the adjacency matrix of a single edge drawn from $f$. Under the assumption that $E(B|f)$ is rank $K$, there exist matrices $X \in \mathds{R}^{n \times K}$ and $S \in \mathds{R}^{K \times K}$ that are a function of $f$ and give the eigendecomposition  $E(B|f) = X S X^T$.  This implies that $\mathbb{P}(e_1 = (i,j) | f) \propto \langle x_i,x_j \rangle _S$, where $x_i$ is the $i$th row of $X$.

\subsection{\texttt{fastRG} samples from a generalization of the RDPG}
\label{gRPG}

Under the RDPG as described in \cite{Young2007}, the expectation of the adjacency matrix is $XX^T$ for some matrix $X \in \mathds{R}^{n \times K}$.  This implies that the expected adjacency matrix is always non-negative definite (i.e. its eigenvalues are non-negative).  However,  some parameterizations of the SBM (and other blockmodels) lead to an expected adjacency matrix with negative eigenvalues (i.e. it is not non-negative definite); for example, if the off-diagonal elements of $S$ are larger than the diagonal elements, then $XSX^T$ could have negative eigenvalues.  Moreover, even if the elements of $X$ and $S$ are positive, as is the case for the low rank models in Table \ref{tab:blockmodels} and as is required for \texttt{fastRG}, it is still possible for $XSX^T$ to have negative eigenvalues.
By modifying the RDPG to incorporate a matrix $S$, the model class below incorporates all types of blockmodels.

\begin{definition}
\label{def:gRPG}
[Generalized Random Product Graph (gRPG) model] 
For $n$ nodes in $K$ dimensions, the gRPG is parameterized by $X \in \mathds{R}^{n \times K}$ and $S \in \mathds{R}^{K \times K}$, where each node $i$ is assigned the $i$th row of $X$, $x_i = (X_{i1},\ \dots,\ X_{iK})^T \in \mathds{R}^K$.  
For $i,j \in V$, define 
\[\lambda_{ij} = \langle x_i,x_j \rangle _S = \sum_k^K \sum_l^K X_{ik} S_{kl} X_{jl}.\]
Under the gRPG, the adjacency matrix $A\in \mathds{R}^{n\times n}$ contains independent elements and the distribution of $A_{ij}$ (i.e. the number of edge from $i$ to $j$) is fully parameterized by $f(\lambda_{ij})$, where $f$ is some mean function.
\end{definition}
Below, we will use the fact that the gRPG only requires that the $\lambda_{ij}$s specify the distribution of $A_{ij}$, allowing for $A_{ij}$ to be non-binary (as in multi-graphs and  weighted graphs) or to have edge probabilities which are a function of $\lambda_{ij}$.


\begin{theorem} 
For $X \in \mathds{R}^{n \times K}$ and $S\in \mathds{R}^{K \times K}$, each with non-negative elements, if $\tilde A$ is the adjacency matrix of a graph sampled with \texttt{fastRG}($X,S$), then $\tilde A$ is a Poisson gRPG with $\tilde A_{ij} \sim Poisson(\langle x_i,x_j \rangle _S )$.
\label{theoremRDPG}
\end{theorem}
\noindent The proof is contained in the appendix.

\begin{remark}[Simulating an undirected graph]
\label{remark:undirected}
As defined, both the gRPG model and \texttt{fastRG} generate directed graphs. An ``undirected gRPG'' should add a constraint to Definition \ref{def:gRPG} that $A_{ij} = A_{ji}$ for all $i,j$.  To sample such a graph with \texttt{fastRG}, input $S/2$ instead of $S$, then after sampling a directed graph with \texttt{fastRG}, symmetrize each edge by removing its direction (this doubles the probability of an edge, hence the need to input $S/2$).  Theorem \ref{theoremRDPG} can be easily extended to show this is an undirected gRPG.
\end{remark}

\begin{remark}[Simulating a graph without self-loops]
\label{remark:selfloops}
As defined, both the gRPG model and \texttt{fastRG} generate graphs with self-loops. A ``gRPG without self-loops'' should add a constraint to Definition \ref{def:gRPG} that $A_{ii} = 0$ for all $i$.  A graph from \texttt{fastRG} can be converted to a gRPG without self-loops by simply (1) sampling $m \sim Poisson(\sum_{u,v} \tilde S_{uv} - \sum_i \langle x_i, x_i \rangle_S)$ and (2) resampling any edge that is a self-loop.
The proof of Theorem \ref{theoremRDPG} can be extended to show that this is equivalent.
\end{remark}

\subsection{Approximate Bernoulli-edges}
\label{approximate} 
To create a simple graph with \texttt{fastRG} (i.e. no multiple edges, no self-loops, and undirected), first sample a graph with \texttt{fastRG}.  Then, perform the modifications described in Remarks \ref{remark:undirected} and \ref{remark:selfloops}. Then, keep an edge between $i$ and $j$ if there is at least one edge in the multiple edge graph;  define the threshold function, $t(A_{ij})=\mathds{1}(A_{ij}>0)$, where $t(A)$ applies element-wise. 

If $\tilde A$ is a Poisson gRPG, then $t(\tilde A)$ is a Bernoulli gRPG with mean function $f(\lambda_{ij}) = 1-exp(-\lambda_{ij})$.  
%
%
Let   $B$ be distributed as Bernoulli gRPG($X,S$) with identity mean function, $B_{ij} \sim Bernoulli(\lambda_{ij})$.  Theorem \ref{theoremThresh} shows that in the sparse setting, there is a coupling between $t(\tilde A)$ and $B$ such that $t(\tilde A)$ is approximately equal to $B$.  The theorem is asymptotic in $n$; a superscript of $n$ is suppressed on $\tilde A, B$ and $\lambda$. 




\begin{theorem} \label{theoremThresh}
Let $\tilde A$ be a Poisson gRPG and let $B$ be a Bernoulli gRPG using the same set of $\lambda_{ij}$s, with $\tilde A_{ij} \sim Poisson(\lambda_{ij})$ and ${B}_{ij}\sim Bernoulli(\lambda_{ij})$.  Let $t(\cdot)$ be the thresholding function for $\tilde A$.

Let $\alpha_n$  be a sequence.  If $\lambda_{ij} = O(\alpha_n/n)$ for all $i,j$ and there exists some constant $c>0$ and $N>0$ such that $\sum_{ij} \lambda_{ij} > c \alpha_n n$ for all $n>N$, then there exists a coupling between $t(\tilde A)$ and $B$ such that 
\[\frac{E\|t(\tilde A) - B\|^2_F}{E\|B\|^2_F}  = O(\alpha_n/n).\]
\end{theorem}
\noindent
A proof is contained in the appendix.

For example, in the sparse graph setting where $\lambda_{ij} = O(1/n)$ and $\sum_{ij} \lambda_{ij} = O(n)$, $\alpha_n = 1$.   Under this setting and the coupling defined in the proof, all of the $O(n)$ edges in $t(\tilde A)$ are contained in $B$ and $B$ has an extra $O(1)$ more edges than $t(\tilde A)$. 

\subsection{Implementation of \texttt{fastRG}}\label{sec:implementation}

Code at \url{https://github.com/karlrohe/fastRG} gives an implementation of \texttt{fastRG} in \texttt{R}.  It also provides wrappers  that simulate the SBM, Degree Corrected SBM, Overlapping SBM, and Mixed Membership SBM.  
The code for these models first generates the appropriate $X$ and then calls \texttt{fastRG}.    In order to help control the edge density of the graph, \texttt{fastRG} and its wrappers can be given an additional argument avgDeg.  If avgDeg is given, then the matrix $S$ is scaled so that \texttt{fastRG} simulates a graph with expected average degree equal to avgDeg.  Without this, parameterizations can easily produce very dense graphs.  

To accelerate the running time of \texttt{fastRG}, the implementation is slightly different than the statement of the algorithm above.  The difference can be thought of as sampling all of the $(U,V)$ pairs before sampling any of the $I$s or $J$s. In particular, the implementation samples $\varpi \in R^{K \times K}$ as multinomial$(m, \tilde S / \sum_{u,v} \tilde S_{uv})$.  Then, for each $u\in \{1, \dots K\}$, it samples $\sum_v \varpi_{u,v}$-many $I$s from the distribution $\tilde X_{\cdot u}$.  Similarly,   for each $v\in \{1, \dots K\}$, it samples $\sum_u \varpi_{u,v}$-many $J$s from the distribution $\tilde X_{\cdot v}$. Finally, the indexes are appropriately arranged so that there are $\varpi_{u,v}$-many edges $(I,J)$ where $I \sim X_{\cdot u}$ and $J \sim X_{\cdot v}$.  Recall that the statement of \texttt{fastRG} above allows for $X$ and $Y$, where those matrices can have different numbers of rows and/or columns; the implementation also allows for this.  

Under the SBM, it is possible to use \texttt{fastRG} to sample from the Bernoulli gRPG with the \textit{identity} mean function instead of the mean function $1 - \exp(-\langle x_i, x_j \rangle_S)$ that is created by the thresholding  function $t$ from Section \ref{approximate}.  The wrapper for the SBM does this by first transforming each element of $S$ as $-\ln(1-S_{ij})$ and then calling \texttt{fastRG}. The others models are not amenable to this trick; by default, they sample from the Poisson gRPG with identity mean function.

\subsection{Experiments}
\label{simulation}
%
To examine the running time of \texttt{fastRG}, we simulated a range of different values of $n$ and $E(m)$, where $E(m)$ is the expected number of edges.  In all simulations $X=Y$ and $K = 5$.  The elements of $X$ are independent $Poisson(1)$ random variables and the elements of $S$ are independent $Uniform[0,1]$ random variables.  To specify $E(m)$, the parameter avgDeg is set to $E(m)/n$. The values of $n$ range from $10,000$ to $10,000,000$ and the values of $E(m)$ range from $100,000$ to $100,000,000$.  The graph was taken to be directed, with self-loops and multiple edges.  Moreover, the reported times are only to generate the edge list of the random graph; the edge list is not converted into a sparse adjacency matrix, which in some experiments would have more than doubled the running time.  Each pair of $n$ and $E(m)$ is simulated one time; deviations around the trend lines indicate the variability in run time.

In Figure \ref{fig1}, the vertical axes present the running time in $\texttt{R}$ on a Retina 5K iMac, 27-inch, Late 2014 with 3.5 GHz Intel i5 and 8GB of 1600 MHz DDR3 memory. In the left panel of Figure \ref{fig1}, each line corresponds to a single value of $n$ and $E(m)$ increases along the horizontal axis. In the right panel of Figure \ref{fig1}, each line corresponds to a single value of $E(m)$ and $n$ increases along the horizontal axis.  All axes are on the $log_{10}$ scale.  The solid black line has a slope of 1.  Because the data aligns with this black line, this suggests that \texttt{fastRG} runs in linear time.

The computational bottleneck is sampling the $I$s and $J$s. The implementation uses Walker's Alias Method \citep{walker1977efficient} (via \texttt{sample} in \texttt{R}).  To take $m$ samples from a distribution over $n$ elements, Walker's Alias Method  requires $O(m + \ln(n)n)$ operations  \citep{vose1991linear}.  
 However, the log dependence is not clearly evident in the right plot of Figure \ref{fig1}; perhaps it would be visible for larger values of $n$.

\begin{figure}[!ht]
  \centering
  \subfigure{
    \includegraphics[width=7.2cm]{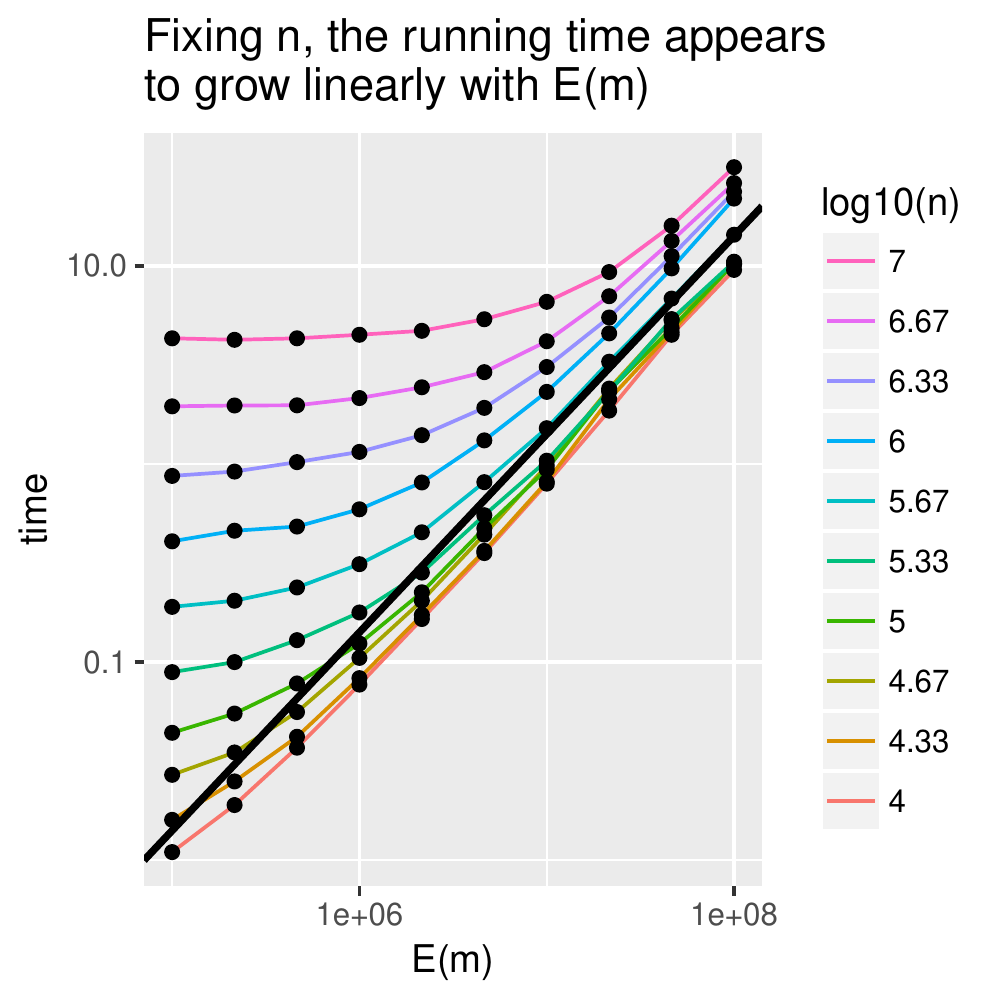}
    }
  \subfigure{
    \includegraphics[width=7.2cm]{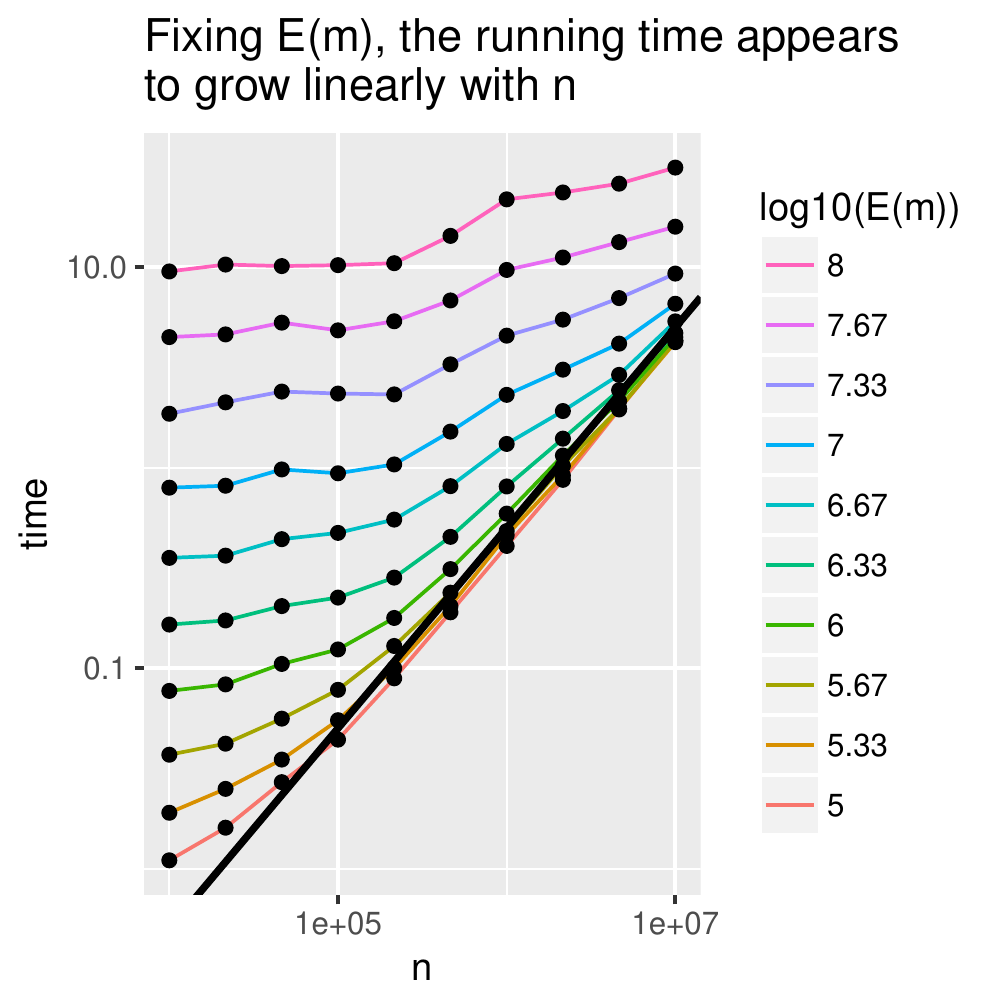}
    }
  \caption{Both plots present the same experimental data.  In the left plot, each  line corresponds to a different value of $n$ and they are presented as a function of $E(m)$.  In the right plot, each  line corresponds to a different value of $E(m)$ and they are presented as a function of $n$.  On the right side of both plots, the lines start to align with the solid black line, suggesting a linear dependence on $E(m)$ and $n$.  
  }
  \label{fig1}
\end{figure}

\vspace{.1in}
\noindent
\textbf{Acknowledgements:} 
This work is supported in part by the U. S. Army Research Office under grant number W911NF1510423, the National Science Foundation under grant number DMS-1612456.
\appendix
  \renewcommand{\appendixname}{Appendix~\Alph{section}}

\section{Proofs}
For an integer $d$, define $1_d \in \mathds{R}^d$ as a vector of ones.  The proof of Theorem \ref{theoremRDPG} requires the following lemma, which says that a vector (or matrix) of independent Poisson entries becomes multinomial when you condition on the sum of the vector (or matrix).
\begin{lemma}
\label{lemma1}
Let $A \in \mathds{R}^{n \times n}$ be the random matrix whose $i,j th$ element $A_{ij} \overset{i.d.}{\sim}$ Pois$(\lambda_{ij})$ $i,j=1, \ \dots, \ n$.
Then conditioned on $\sum_{i,j}A_{ij} =  1_n^TA1_n=m$,
\begin{equation*}
(A_{11},\ A_{12},\ \dots,\ A_{nn}) \sim Multinomial(m,\ \lambda / \sum_{ij} \lambda_{ij})
\end{equation*}
where $\lambda = (\lambda_{11},\ \lambda_{12},\ \dots,\ \lambda_{nn})$.  That is, let $a\in \mathds{R}^{n \times n}$ be a fixed matrix of integers with $1_n^Ta1_n=m$, then
\begin{eqnarray*}
\mathbb{P}(A = a|1_n^T A 1_n = m)&=&\mathbb{P}(A_{11} = a_{11}, \ A_{12} = a_{12},\ \dots,\ A_{nn} = a_{nn}|1_n^T A 1_n = m) \\
&=&  \frac{m!}{\prod_{i,j}a_{ij}!} \prod_{i,j} \left(\frac{\lambda_{ij}}{\lambda_{11} + \lambda_{12}+\cdots +\lambda_{nn}} \right)^{a_{ij}}.
\end{eqnarray*}
\end{lemma}

\vspace{.2 in}

For completeness, a proof of this classical result is given at the end of the paper. The next proof is a proof of Theorem \ref{theoremRDPG}.

\vspace{.2 in}

\begin{proof}
Let $A$ come from the Poisson gRPG with $X$ and $S$ and identity mean function.  Let $\tilde A$ be a sample from \texttt{fastRG}.   For any fixed adjacency matrix $a$, we will show that $\mathbb{P}(A = a)=\mathbb{P}(\tilde A = a)$.

Define $m = 1_n^Ta1_n$ and decompose the probabilities,
\begin{eqnarray}
\mathbb{P}(A = a) &=& \mathbb{P}(1_n^T A 1_n =m) \mathbb{P}(A=a| 1_n^T A 1_n =m) \\
\mathbb{P}(\tilde A = a) &=& \mathbb{P}(1_n^T \tilde A 1_n =m) \mathbb{P}(\tilde A=a| 1_n^T \tilde A 1_n =m).
\end{eqnarray}
The proof will be divided into two parts.  The first part shows that $\mathbb{P}(1_n^T A 1_n =m) = \mathbb{P}(1_n^T \tilde A 1_n =m)$ and the second part will show that  $\mathbb{P}(A=a| 1_n^T A 1_n =m) = \mathbb{P}(\tilde A=a| 1_n^T \tilde A 1_n =m)$.

\textbf{Part 1:} 
Because the sum of independent Poisson variables is still Poisson, $\sum_{ij} A_{ij} \sim Poisson(\sum_{ij}\lambda_{ij}).$  So, we must only show that $1_n^T A 1_n$ and $1_n^T \tilde A 1_n$ have the same Poisson parameter:
\[\sum_{ij}\lambda_{ij} = 1_n^T X S X 1_n  =  1_n^T XC C^{-1} S C^{-1} C X 1_n  =  1_n^T \tilde X\tilde S \tilde X 1_n  =  1_n^T \tilde X\tilde S \tilde X 1_n  =  1_K^T \tilde S 1_K  =  \sum_{u,v} \tilde S_{uv}.\]


\textbf{Part 2:} After conditioning on $1_n^T A 1_n = m$,  Lemma \ref{lemma1} shows that $A$ has the multinomial distribution.  In \texttt{fastRG}, we first sample $1_n^T \tilde A1_n$ and then add edges with the multinomial distribution.  So,  we must only show that the multinomial edge probabilities are equal for $A$ and $\tilde A$.  From Lemma \ref{lemma1}, the multinomial edge probabilities for $A$ are $\lambda_{ij} /\sum_{a,b} \lambda_{ab} $.
%
To compute the multinomial edge probabilities for $\tilde A$, recall that $(I,J)$ is a single edge added to the graph in \texttt{fastRG}. By Theorem \ref{theorem:xlr},
\[\mathbb{P}(\tilde A_{ij} = 1 | 1_n^T\tilde A1_n = 1  ) =  \mathbb{P}\big((I,J)=(i,j)\big) = \frac{\langle x_i, x_j\rangle_S}{\sum_{a,b} \langle x_a, x_b\rangle_S} =  \frac{\lambda_{ij}}{\sum_{a,b} \lambda_{ab}}\]
This concludes the proof.
\end{proof}

\begin{proof}[Proof of Theorem \ref{theoremThresh}]
Let $U_{ij} \overset{i.i.d}{\sim}$ $Uniform(0,1)$. Define $\mathcal{A}$ and $\mathcal{B}$: 
\begin{align*}
\mathcal{A}_{ij} &= \bold{1}(U_{ij}>1-e^{-\lambda_{ij}}),\\
\mathcal{B}_{ij} &= \bold{1}(U_{ij}>\lambda_{ij}).
\end{align*}
Note that $\mathcal{A}$ and $\mathcal{B}$ are equal in distribution to $ t(\tilde A)$ and $B$ respectively.  By Taylor expansion,
\[E\|\mathcal{A} - \mathcal{B}\|^2_F = \sum_{i,j}(\lambda_{ij}-(1-e^{-\lambda_{ij}})) =  \sum_{i,j} \sum_{k=2}^\infty (-\lambda_{ij})^k/k! =  \sum_{i,j} O(\lambda_{ij}^2) =  \sum_{i,j} O((\alpha_n/n)^2) =  O(\alpha_n^2).\]
Then, $E\|B\|_F^2  = \sum_{ij} \lambda_{ij} > c \alpha_n n $.  So, defining $t(\tilde A)$ and $B$ with the above coupling yields the result.

\end{proof}

%
%
%
%
%
%

\begin{proof}[proof of Lemma 1]
\begin{align*}
\mathbb{P}(A=a|1_n^T A 1_n = m) &= \frac{\mathbb{P}(A=a)}{\mathbb{P}(1_n^T A 1_n = m)}
 = \dfrac{\prod_{i,j} \dfrac{\lambda_{ij}^{a_{ij}}}{a_{ij}!}e^{-\lambda_{ij}}}{\dfrac{(\lambda_{11} + \lambda_{12}+\cdots +\lambda_{nn})^m}{m!}e^{-({\lambda_{11} + \lambda_{12}+\cdots +\lambda_{nn})}}}\\
& = \frac{m!}{\prod_{i,j}a_{ij}!} \prod_{i,j} \left(\frac{\lambda_{ij}}{\lambda_{11} + \lambda_{12}+\cdots +\lambda_{nn}} \right)^{a_{ij}}
\end{align*}
\end{proof}

\bibliographystyle{plainnat}
\bibliography{bibfile}

\end{document}